\documentclass[11pt, onecolumn]{IEEEtran}
\usepackage[letterpaper, left=1in, right=1in, bottom=1in, top=1in]{geometry}
\ifCLASSINFOpdf
\else
\fi

\usepackage[english]{babel}
\usepackage{enumitem}
\usepackage{enumerate}
\usepackage{graphicx}
\usepackage{color}
\usepackage[dvipsnames]{xcolor}
\usepackage[T1]{fontenc}
\usepackage{subcaption}
\usepackage{dsfont}
\usepackage{amsfonts}
\usepackage{graphicx,wrapfig}
\usepackage[T1]{fontenc}
\usepackage{amsmath}
\usepackage{mathtools, cuted}
\usepackage{amsthm}
\usepackage{amstext}
\usepackage{amssymb}
\usepackage{mathrsfs}
\usepackage[square,numbers, sort&compress]{natbib}
\usepackage{mathtools}
\usepackage{tikz}
\usepackage{marginnote}
\usepackage{tkz-tab}
\usetikzlibrary{topaths,calc}
\usepackage{pgfplots}
\usepackage{stackengine}
\usepackage{algorithmic}
\usepackage[ruled]{algorithm} 

\usepackage[colorinlistoftodos]{todonotes}

\usetikzlibrary{arrows,positioning, shapes}
\usepackage{float}

\usepackage{setspace}
\usepackage{hyperref}
\hypersetup{colorlinks,
   colorlinks,
    citecolor=red,
    filecolor=green,
  linkcolor=blue,
    linktoc=page,
   urlcolor=blue
}

\def\R{\mathbb{R}}

\def\argmin{\mathop{\rm arg\, min}}
\def\eps{\varepsilon}

\def\E{{\mathbb E}}

\def\P{{\mathcal P}}
\def\Q{{\mathcal Q}}
\def\T{{\mathcal T}}
\def\X{{\mathcal X}}

\def\L{{\mathcal L}}

\def\Z{{\mathcal Z}}
\def\U{{\mathcal U}}
\def\V{{\mathcal V}}

\def\sBer{{\mathsf{Bernoulli}}}

\def\sK{{\mathsf K}}

\def\sE{{\mathsf E}}

\def\tv{{\mathsf {TV}}}
\def\kl{{\mathsf {KL}}}

\usepackage[font=small,labelsep=space]{caption}
\DeclareCaptionLabelSeparator{dot}{.~}
\newcommand{\eq}[1]{\begin{equation*}
#1
\end{equation*}}
\newcommand{\eqn}[2]{\begin{equation}
\label{#1}
#2
\end{equation}}
\newcommand{\al}[1]{\begin{align*}
#1
\end{align*}}
\newcommand{\aln}[1]{\begin{align}
#1
\end{align}}

\newcounter{example}
\newenvironment{example}[1][]{\refstepcounter{example}\par\medskip
   \noindent \textit{Example~\theexample. #1} \rmfamily}{\medskip}
\newtheorem{definition}{Definition}
\newtheorem{theorem}{Theorem}
\newtheorem{corollary}{Corollary}

\newtheorem{lemma}{Lemma}
\newtheorem{remark}{Remark}

\newcommand{\markov}{\mathrel\multimap\joinrel\mathrel-%
\mspace{-9mu}\joinrel\mathrel-}

\usepgfplotslibrary{fillbetween}
\usepackage{times}
\usepackage{tikz}
\usepackage{amsmath, bm}
\usepackage{verbatim}
\usetikzlibrary{arrows,shapes}
\tikzstyle{RectObject}=[rectangle,fill=white,draw,line width=0.2mm]
\tikzstyle{line}=[draw]
\tikzstyle{arrow}=[draw, -latex]
\usetikzlibrary{decorations.pathmorphing}
\usetikzlibrary{calc,shapes, positioning}
\usepackage{graphicx}
\usetikzlibrary{shapes.geometric}

\DeclareFontFamily{U}{BOONDOX-calo}{\skewchar\font=45 }
\DeclareFontShape{U}{BOONDOX-calo}{m}{n}{
	<-> s*[1.05] BOONDOX-r-calo}{}
\DeclareFontShape{U}{BOONDOX-calo}{b}{n}{
	<-> s*[1.05] BOONDOX-b-calo}{}
\DeclareMathAlphabet{\mathcalboondox}{U}{BOONDOX-calo}{m}{n}
\SetMathAlphabet{\mathcalboondox}{bold}{U}{BOONDOX-calo}{b}{n}
\DeclareMathAlphabet{\mathbcalboondox}{U}{BOONDOX-calo}{b}{n}

\definecolor{DukeBlue}{HTML}{001A57}
\definecolor{DarkRed}{rgb}{0.75, 0.0, 0.0}
\definecolor{DarkGreen}{rgb}{0.0, 0.5, 0.0}

\usepackage[normalem]{ulem}

\allowdisplaybreaks

\author{}
\date{}
\begin{document}

\title{\vspace{5.5mm}Local Differential Privacy Is Equivalent to Contraction of $\sE_\gamma$-Divergence}
\author{%
 Shahab Asoodeh${}^\dagger$, Maryam Aliakbarpour${}^*$, and Flavio P. Calmon${}^\dagger$ \\
 \small ${}^\dagger$Harvard University, ${}^*$University of Massachusetts Amherst}

	\maketitle
\begin{abstract}
We investigate the local differential privacy (LDP) guarantees of a randomized privacy mechanism via its contraction properties. We first show that LDP constraints  can be equivalently cast in terms of the contraction coefficient of the $\sE_\gamma$-divergence. We then use this equivalent formula to express  LDP guarantees of privacy mechanisms in terms of contraction coefficients of arbitrary $f$-divergences. When combined with standard estimation-theoretic tools (such as Le Cam's and Fano's converse methods), this result allows us to study the trade-off between privacy and utility in several testing and minimax and Bayesian estimation problems.   
\end{abstract}

\section{Introduction}
A major challenge in modern machine learning applications is balancing   statistical efficiency with the privacy of individuals from whom data is obtained. In such applications, privacy is often quantified in terms of Differential Privacy (DP) \cite{Dwork_Calibration}. DP has several variants, including approximate DP \cite{Dwork-OurData}, Rényi DP \cite{RenyiDP}, and others \cite{ZeroDP, Concentrated_Dwork, GaussianDP, asoodeh2020variants}. Arguably, the most stringent flavor of DP is \textit{local differential privacy}  (LDP) \cite{evfimievski2003limiting, Shiva_subsampling}. Intuitively, a randomized mechanism (or a Markov kernel) is said to be locally differentially private  if its output does not vary significantly with arbitrary perturbations of the input.  

More precisely, a mechanism is said to be $\eps$-LDP (or \textit{pure} LDP) if the privacy loss random variable, defined as the log-likelihood ratio of the output for any two different inputs, is smaller than $\eps$ with probability one.  One can also consider an \textit{approximate} variant of this constraint: $\sK$ is said to be $(\eps, \delta)$-LDP if the privacy loss random variable does not exceed $\eps$ with probability at least $1-\delta$ (see Def.~\ref{Def:LDP} for the formal definition).

The study of statistical efficiency under LDP constraints has gained considerable traction, e.g., \cite{Shiva_subsampling, Duchi_LDP_MinimaxRates,LDP_gaboardi19a, Duchi_Federatedprotection, kairouz2014extremal_JMLR,LDP_Fisher, LDP_Acharya1, LDP_DistributionEstimation, LDP_LinearRegression,rohde2020, evfimievski2003limiting}.  Almost all of these works consider $\eps$-LDP and provide meaningful bounds
only for sufficiently small values of $\eps$ (i.e., the high privacy regime). 
For instance, Duchi et al.\ \cite{Duchi_LDP_MinimaxRates} studied minimax estimation problems under $\eps$-LDP constraints and showed that for $\eps\leq 1$, the price of privacy is to reduce the effective sample size from $n$ to $\eps^2 n$. A slightly improved version of this result appeared in \cite{kairouz16_LDPEstimation, kairouz2014extremal_JMLR}.  
More recently, Duchi and Rogers \cite{Duchi_interactivity}
developed a framework based on the \textit{strong} data processing inequality (SDPI) \cite{ahlswede1976} and derived lower bounds for minimax estimation risk under $\eps$-LDP that hold for any $\eps\geq 0$.  

In this work, we develop an SDPI-based  framework for studying hypothesis testing and estimation problems under $(\eps,\delta)$-LDP, extending the results of \cite{Duchi_interactivity} to  approximate LDP. In particular, we derive bounds for both the minimax and Bayesian  estimation risks that hold for any $\eps\geq 0$ and $\delta\geq 0$. Interestingly, when setting $\delta = 0$, our bounds can be slightly stronger than \cite{Duchi_LDP_MinimaxRates}.

Our main mathematical tool is   an equivalent expression for DP in terms of 
$\sE_\gamma$-divergence. 
Given $\gamma\geq 1$, the $\sE_\gamma$-divergence between two distributions $P$ and $Q$ is defined as 
\begin{equation}\label{Def:E_Gamma}
 \sE_\gamma(P\|Q) \coloneqq   \frac{1}{2}\int|\text{d}P-\gamma \text{d}Q| - \frac{1}{2}(\gamma-1).  
\end{equation}
We show that a mechanism $\sK$ is $(\eps, \delta)$-LDP \textit{if and only if} $$\sE_\gamma(P\sK\|Q\sK)\leq \delta \sE_\gamma(P\|Q)$$ for $\gamma = e^\eps$ and any pairs of distributions $(P, Q)$ where $P\sK$ represents the output distribution of $\sK$ when the input distribution is $P$. Thus, the approximate LDP guarantee of a mechanism can be fully characterized by its contraction under $\sE_\gamma$-divergence.  When combined with
standard statistical techniques, including Le Cam's and Fano's methods \cite{Yu1997, Tsybakov_Book}, $\sE_\gamma$-contraction leads to general lower bounds for the minimax and Bayesian risk under $(\eps, \delta)$-LDP for any $\eps\geq 0$ and $\delta\in [0,1]$. In particular, we show that the price of privacy in this case is to reduce the sample size from $n$ to $n[1-e^{-\eps}(1-\delta)]$. 

There exists several results connecting pure LDP to  the contraction properties of KL divergence $D_\kl$ and total variation distance $\tv$. For instance, for any $\eps$-LDP mechanism $\sK$, it is shown in \cite[Theorem 1]{Duchi_LDP_MinimaxRates} that $D_\kl(P\sK\|Q\sK)\leq 2(e^\eps-1)^2\tv^2(P, Q)$ and in \cite[Theorem 6]{kairouz2014extremal_JMLR} that $\tv(P\sK\|Q\sK)\leq \frac{e^\eps-1}{e^\eps + 1}\tv(P, Q)$ for any pairs $(P, Q)$. 
Inspired by these results,   we further show that if $\sK$ is $(\eps, \delta)$-LDP then $D_f(P\sK\|Q\sK)\leq [1-e^{-\eps}(1-\delta)] D_f(P\|Q)$ for any \textit{arbitrary} $f$-divergences $D_f$ and any pairs $(P, Q)$. 

\noindent\textbf{Notation.} For a random variable $X$, we write $P_X$ and $\X$ for its distribution (i.e., $X\sim P_X$) and its alphabet, respectively. 
For any set $A$, we denote by $\P(A)$ the set of all probability distributions on $A$.    
Given two sets $\X$ and $\Z$, a Markov kernel (i.e., channel) $\sK$ is a mapping from $\X$ to $\P(\Z)$ given by $x\mapsto \sK(\cdot|x)$.
Given $P\in \P(\X)$ and a Markov kernel $\sK:\X\to \P(\Z)$, we let $P\sK$ denote the output distribution of $\sK$ when the input distribution is $P$, i.e., $P\sK(\cdot) = \int\sK(\cdot|x)P(\text{d}x)$. Also, we use $\mathsf{BSC}(\omega)$ to denote the binary symmetric channel with crossover probability $\omega$. For sequences $\{a_n\}$ and $\{b_n\}$, we use $a_n\gtrsim b_n$ to indicate $a_n\geq C b_n$ for some universal constant $C$.

\section{Preliminaries}\label{Sec:Background}
\subsection{$f$-Divergences} 
Given a convex function $f:(0,\infty)\to\mathbb{R}$ such that $f(1)=0$, the $f$-divergence between two probability measures $P\ll Q$ is defined as \cite{Csiszar67, Ali1966AGC}
\eqn{}{D_f(P\|Q)\coloneqq \E_Q\Big[f\big(\frac{\textnormal{d}P}{\textnormal{d}Q}\big)\Big].}
Due to convexity of $f$, we have $D_f(P\|Q)\geq f(1) = 0$. If, furthermore, $f$ is strictly convex at $1$, then equality holds if and only $P=Q$. Popular examples of $f$-divergences include $f(t) = t\log t$ corresponding to KL divergence, $f(t) = |t-1|$ corresponding to total variation distance, and $f(t) = t^2-1$ corresponding to $\chi^2$-divergence. 
In this paper, we mostly concern with an important  sub-family of $f$-divergences
associated with $f_\gamma(t) = \max\{t-\gamma, 0\}$
for a parameter $\gamma\geq 1$. 
The corresponding $f$-divergence, denoted by $\sE_\gamma(P\|Q)$, is called $\sE_\gamma$-divergence (or sometimes {\em hockey-stick divergence} \cite{hockey_stick}) and is explicitly defined in \eqref{Def:E_Gamma}. It  appeared in \cite{polyanskiy2010channel} for proving converse channel coding results and also used in \cite{Balle:Subsampling, Balle2019mixing, asoodeh2020online, asoodeh2020variants} for characterizing privacy guarantees of iterative algorithms in terms of other variants of DP. 

\subsection{Contraction Coefficient}

All $f$-divergences satisfy data processing inequality, i.e., $D_f(P\sK\|Q\sK)\leq D_f(P\|Q)$ for any pair of probability distributions $(P,Q)$ and Markov kernel $\sK$ \cite{Csiszar67}. However, in many cases, this inequality is strict. The \textit{contraction coefficient} of Markov kernel $\sK$ under $D_f$-divergence $\eta_f(\sK)$ is the smallest number $\eta\leq 1$ such that $D_f(P\sK\|Q\sK)\leq \eta D_f(P\|Q)$ for any pair of probability distributions $(P,Q)$. Formally, $\eta_f(\sK)$ is defined as 
\begin{equation}\label{Eq:SDPI_f_Div}
    \eta_f(\sK)\coloneqq \sup_{\substack{P,Q\in \P(\X):\\ D_f(P\|Q)\neq 0}}\frac{D_f(P \sK\|Q\sK)}{D_f(P\|Q)}.
\end{equation}
Contraction coefficients have been  studied for several $f$-divergences, e.g., 
$\eta_\tv$ for total variation distance was studied in \cite{Dobrushin, Del_Moral_Contraction, COHEN_Paper}, $\eta_\kl$ for $\mathsf{KL}$-divergence  was studied in \cite{Anantharam_SDPI, Polyankiy_SDPI_Networks, Yury_Dissipation, Flavio_Polyankiy, Makur_SDPIjournal, Raginsky_SDPI}, and $\eta_{\chi^2}$ for $\chi^2$-divergence was studied in \cite{COHEN_Paper,Raginsky_SDPI, Witsenhausen_MC}.
In particular,  
Dobrushin \cite{Dobrushin} showed that $\eta_{\mathsf{TV}}$ 
has a remarkably simple two-point characterization  $\eta_\mathsf{TV}(\sK) = \sup_{x_1,x_2\in\X} \mathsf{TV}(\sK(\cdot|x_1), \sK(\cdot|x_2))$.


Similarly, one can plug $\sE_\gamma$-divergence into \eqref{Eq:SDPI_f_Div} and define the contraction coefficient $\eta_\gamma(\sK)$ for a Markov kernel $\sK$ under $\sE_\gamma$-divergence. 
This contraction coefficient has recently been studied in \cite{asoodeh2020online} for deriving approximate DP guarantees for online algorithms. In particular, it was shown \cite[Theorem 3]{asoodeh2020online} that $\eta_\gamma$ enjoys a simple two-point characterization, i.e.,   $\eta_\gamma(\sK) = \sup_{x_1,x_2\in\X} \sE_\gamma(\sK(\cdot|x_1)\| \sK(\cdot|x_2))$. Since $\sE_1(P\|Q) = \tv(P, Q)$, this is a natural extension of Dobrushin's result.

\subsection{Local Differential Privacy}
Suppose $\sK$ is a randomized mechanism  mapping each $x\in \X$ to a distribution $\sK(\cdot|x)\in \P(\Z)$.
One could view $\sK$ as a Markov kernel (i.e., channel) $\sK:\X\to \P(Z)$.
\begin{definition}[\cite{evfimievski2003limiting,Shiva_subsampling}]\label{Def:LDP}
A mechanism $\sK:\X\to \P(\Z)$ is $(\eps, \delta)$-LDP for $\eps\geq 0$ and $\delta\in [0,1]$ if 
\eqn{}{\sup_{x, x'\in \X} \sup_{A\subset \Z}~\left[\sK(A|x)-e^\eps\sK(A|x')\right] \leq \delta.}
$\sK$ is said to be $\eps$-LDP if it is $(\eps, 0)$-LDP.
Let $\Q_{\eps,\delta}$ be the collection of all Markov kernels $\sK$ with the above property. When $\delta = 0$, we use $\Q_{\eps}$ to denote $\Q_{\eps,0}$. 
\end{definition}
\noindent \textbf{Interactivity in Privacy-Preserving Mechanisms:}
Suppose there are $n$ users, each in possession of a datapoint $X_i$, $i\in [n]\coloneqq \{1, \dots, n\}$. The users wish to apply a mechanism $\sK_i$ that generates  a privatized version of $X_i$, denoted by $Z_i$. We say that the collection of mechanisms $\{\sK_i\}$ is  \textit{non-interactive} if $\sK_i$ is entirely determined by $X_i$ and independent of $(X_j,Z_j)$ for $j\neq i$. When all users apply the same mechanism $\sK$,  we can view $Z^n \coloneqq (Z_1, \ldots, Z_n)$ as  independent applications of $\sK$ to each $X_i$. We denote this overall mechanism by $\sK^{\otimes n}$. 
If interactions between users are permitted, then  $\sK_i$ need not depend only on $X_i$. In this case, we denote the overall mechanism $\{\sK_i\}_{i=1}^n$ by $\sK^n$.  In particular, the \textit{sequentially interactive} \cite{Duchi_LDP_MinimaxRates} setting refers to the case when  the input of $\sK_i$  depends on both $X_i$ and the outputs $Z^{i-1}$ of the $(i-1)$ previous mechanisms.

\section{LDP As the Contraction of $\sE_\gamma$-Divergence}

We show next that the $(\eps, \delta)$-LDP constraint, with $\delta$ not necessarily equal to zero, is \textit{equivalent} to the contraction of $\sE_\gamma$-divergence.  
\begin{theorem}\label{thm:LDP_Contraction}
A mechanism $\sK$ is $(\eps, \delta)$-LDP if and only if  $\eta_{e^\eps}(\sK)\leq \delta$ or equivalently
\eq{\sK\in \Q_{\eps, \delta} ~\Longleftrightarrow ~\sE_{e^\eps}(P\sK\|Q\sK)\leq \delta \sE_{e^\eps}(P\|Q), \quad \forall P, Q.} 
\end{theorem}
We note that
Duchi et al. \cite{Duchi_LDP_MinimaxRates} showed that if $\sK$ is $\eps$-LDP then $D_{\kl}(P\sK\|Q\sK)\leq 2 (e^\eps-1)^2\tv^2(P, Q)$. They then informally concluded from this result that $\eps$-LDP acts as a contraction on the space of probability measures. Theorem \ref{thm:LDP_Contraction} makes this observation precise.

According to Theorem \ref{thm:LDP_Contraction}, a mechanism $\sK$ is $\eps$-LDP if and only if $\sE_{e^{\eps}}(P\sK\|Q\sK) =0$ for any distributions $P$ and $Q$.  
An example of such Markov kernels is given next.
\begin{example}[(Randomized response mechanism)]\label{example:RRMechanism}
Let $\X=\Z=\{0,1\}$ and consider the mechanism given by the binary symmetric channel $\mathsf{BSC}(\omega_\eps)$ with $\omega_\eps\coloneqq \frac{1}{1+e^\eps}$. This is often called randomized response mechanism \cite{warner1965randomized}  and denoted by $\sK^{\eps}_{\mathsf{RR}}$. 
This simple mechanism is well-known to be $\eps$-LDP which can now be verified via Theorem~\ref{thm:LDP_Contraction}. Let $P = \sBer(p)$ and $Q=\sBer(q)$ with $p,q\in [0,1]$. Then $P\sK^\eps_{\mathsf{RR}} = \sBer(p*\omega_\eps)$ and $P\sK^\eps_{\mathsf{RR}} = \sBer(q*\omega_\eps)$ where $a*b \coloneqq a(1-b) + b(1-a)$. It is straightforward to verify that $|p*\omega_\eps-e^\eps q*\omega_\eps|+|1-p*\omega_\eps-e^\eps (1-q*\omega_\eps)| = 0.5(e^\eps-1)$ for any $p,q$, implying $\sE_{e^\eps}(P\sK^\eps_{\mathsf{RR}}\|Q\sK^\eps_{\mathsf{RR}}) = 0$.
When $|\X| = k\geq 2$, a simple generalization of this mechanism, called $k$-ary randomized response, has been reported in literature (see, e.g., \cite{kairouz16_LDPEstimation, kairouz2014extremal_JMLR}) and is defined by $\Z = \X$ and $\sK_{\mathsf{kRR}}(x|x) = \frac{e^\eps}{k-1+e^\eps}$ and $\sK_{\mathsf{kRR}}(z|x) = \frac{1}{k-1+e^\eps}$ for $z\neq x$. Again, it can be verified that for this mechanism we have $\sE_{e^\eps}(P\sK^\eps_{\mathsf{kRR}}\|Q\sK^\eps_{\mathsf{kRR}}) = 0 $, for all Bernoulli $P$ and $Q$. 
\end{example}


$\sE_\gamma$-divergence underlies all other $f$-divergences, in a sense that any arbitrary $f$-divergence can be represented by $\sE_\gamma$-divergence \cite[Corollary 3.7]{cohen1998comparisons}. Thus, an LDP constraint  implies that a Markov kernel contracts for \textit{all} $f$-divergences, in a similar spirit to   $\sE_\gamma$-contraction in Theorem~\ref{thm:LDP_Contraction}.  
\begin{lemma}\label{Lemma:UB_f_Div}
Let $\sK\in \Q_{\eps, \delta}$ and $\varphi(\eps, \delta)\coloneqq 1-(1-\delta)e^{-\eps}$. Then,  $\eta_f(\sK)\leq \varphi(\eps, \delta)$ or, equivalently, 
$$D_f(P\sK\|Q\sK)\leq D_f(P\|Q) \varphi(\eps, \delta) \qquad \forall P, Q\in \P(\X).$$
\end{lemma}
Notice that this lemma holds for any $f$-divergences and any  \textit{general} family of $(\eps, \delta)$-LDP mechanisms. However, it can be improved if one considers particular mechanisms or a certain $f$-divergence. For instance, it is known that $\eta_\kl(\mathsf{BSC}(\omega)) = (1-2\omega^2)$ \cite{ahlswede1976}. Thus, we have $\eta_\kl(\sK^\eps_{\mathsf{RR}}) = (\frac{e^\eps-1}{e^\eps+1})^2$ for the randomized response mechanism $\sK^\eps_{\mathsf{RR}}$ (cf. Example~\ref{example:RRMechanism}), while  Lemma~\ref{Lemma:UB_f_Div} implies that $\eta_\kl(\sK^\eps_{\mathsf{RR}})\leq 1-e^{-\eps}$. Unfortunately, $\eta_\kl$ is difficult to compute in closed form  for general Markov kernels, in which case   Lemma~\ref{Lemma:UB_f_Div} provides a useful alternative. 


Next, we  extend Lemma~\ref{Lemma:UB_f_Div} for the non-interactive  mechanism. 
Fix an $(\eps, \delta)$-LDP mechanism $\sK$ and consider the corresponding non-interactive mechanism $\sK^{\otimes n}$. To obtain upper bounds on $\eta_f(\sK^{\otimes n})$ directly through Lemma~\ref{Lemma:UB_f_Div}, we would first need to derive privacy parameters of $\sK^{\otimes n}$ in terms of $\eps$ and $\delta$ (e.g., by applying composition theorems).  Instead, we can use the tensorization properties of contraction coefficients (see, e.g., \cite{Raginsky_SDPI,Makur_SDPIjournal}) to relate $\eta_f(\sK^{\otimes n})$ to $\eta_f(\sK)$ and then apply Lemma~\ref{Lemma:UB_f_Div}, as described next. 

\begin{lemma}\label{lemma:UB_Eta_Kn}
Let $\sK\in \Q_{\eps, \delta}$ and $\varphi_n(\eps, \delta)\coloneqq 1-e^{-n\eps}(1-\delta)^n$. Then $\eta_f(\sK^{\otimes n})\leq \varphi_n(\eps, \delta)$ for  $n\geq 1.$
\end{lemma}
Each of the next three sections provide a different application of the contraction characterization of LDP.
\section{Private Minimax Risk}\label{Sec:Applications}
Let $X^n = (X_1,\dots, X_n)$ be $n$ independent and identically distributed (i.i.d.)\ samples drawn from a distribution $P$ in a family $\P \subseteq \P(\X)$. Let also $\theta:\P\to \T$ be a parameter of a distribution that we wish to estimate. Each user has a sample $X_i$ and applies a privacy-preserving mechanism $\sK_i$ to obtain $Z_i$. Generally, we can assume that $\sK_i$ are sequentially interactive. Given the sequences $\{Z_i\}_{i=1}^n$, the goal is to  estimate $\theta(P)$ through an estimator  $\Psi:\Z^n\to \T$.  
The quality of such estimator is assessed by a semi-metric $\ell:\T\times\T\to \R_+$ and is used to define the minimax risk as: 
\eqn{}{\mathcal R_n(\P, \ell, \eps, \delta)\coloneqq \inf_{\sK^n\subset \Q_{\eps, \delta}}\inf_{\Psi}\sup_{P\in \P}\E[\ell(\Psi(Z^n), \theta(P))].}

The quantity $R_n(\P, \ell, \eps, \delta)$ uniformly characterizes the optimal rate of private statistical estimation over the family $\P$ using the best possible estimator and privacy-preserving mechanisms in $\Q_{\eps, \delta}$. 
In the absence of privacy constraints (i.e., $Z^n = X^n$), we denote the minimax risk by $\mathcal R_n(\P, \ell)$. 

The  first step in deriving information-theoretic lower bounds for minimax risk is to reduce the above estimation problem to a testing problem \cite{Tsybakov_Book, Barron_ITMinimax, Yu1997}. To do so, we need to construct an index set $\V$ with $|\V|<\infty$ and a family of distributions $\{P_v, v\in \V\}\subseteq \P$ such that $\ell(\theta(P_v), \theta(P_{v'}))\geq 2\tau$ for all $v \neq v'$ in $\V$ for some $\tau>0$. The canonical testing problem is then defined as follows: Nature chooses a random variable $V$ uniformly at random from $\V$, and then conditioned on $V= v$, the samples $X^n$ are drawn i.i.d.\ from $P_v$, denoted by $X^n\sim P^{\otimes n}_v$. Each $X_i$ is then fed to a mechanism $\sK_i$ to generate $Z_i$. It is well-known \cite{Yu1997, Barron_ITMinimax, Tsybakov_Book} that $\mathcal R_n(\P, \ell)\geq \tau \mathsf{P}_{\mathsf e}(V|X^n)$, where $\mathsf{P}_{\mathsf e}(V|X^n)$ denotes the probability of error in guessing $V$ given $X^n$. 
Replacing $X^n$ by its $(\eps, \delta)$-privatized samples $Z^n$ in this result, one can obtain a lower bound on $R_n(\P, \ell, \eps, \delta)$ in terms of $\mathsf{P}_{\mathsf e}(V|Z^n)$. Hence, the remaining challenge is to lower-bound  $\mathsf{P}_{\mathsf e}(V|Z^n)$ over the choice of mechanisms $\{\sK_i\}$.
There are numerous techniques for this objective depending on $\V$. We focus on two such approaches, namely Le Cam's and Fano's method, that bound $\mathsf{P}_{\mathsf e}(V|Z^n)$ in terms of total variation distance and mutual information and hence allow us to invoke Lemmas~\ref{Lemma:UB_f_Div} and \ref{lemma:UB_Eta_Kn}. 

\subsection{Locally Private Le Cam's Method} 
Le Cam's method is applicable when $V$ is a binary set and contains, say,  $P_0$ and $P_1$. In its simplest form, it relies on the  inequality (see \cite[Lemma 1]{Yu1997} or \cite[Theorem 2.2]{Tsybakov_Book}) $\mathsf{P}_{\mathsf{e}}(V|X^n) \geq \frac{1}{2} \left[1-\tv(P^{\otimes n}_0, P^{\otimes n}_1)\right]$. Thus, it yields the following lower bound for non-private minimax risk
\aln{
\mathcal R_n(\P, \ell)&\geq \frac{\tau}{2} \left[1-\tv(P^{\otimes n}_0, P^{\otimes n}_1)\right]\label{LeCam_TV}\\
& \geq \frac{\tau}{2}\left[1-\frac{1}{\sqrt{2}}\sqrt{nD_\kl(P_0\| P_1)}\right], \label{LeCam_KL}}
for any $P_0\neq P_1$ in $\P$, where the second inequality follows from Pinsker's inequality and chain rule of KL divergence. 
In the presence of privacy, the estimator $\Psi$ depends on $Z^n$ instead of $X^n$, which is generated by a sequentially interactive mechanism $\sK^n$. To write the private counterpart of \eqref{LeCam_TV}, we need to replace $P^{\otimes n}_0$ and $P^{\otimes n}_1$ with $P^{\otimes n}_0\sK^n$ and $P^{\otimes n}_1\sK^n$ the corresponding marginals of $Z^n$, respectively.
A lower bound for $\mathcal R_n(\P, \ell, \eps, \delta)$ is therefore obtained  by deriving an upper bound for 
$\tv(P_0^{\otimes n}\sK^n,P_1^{\otimes n}\sK^n)$ for all $\sK^n\subset \Q_{\eps, \delta}$. 
\begin{lemma}\label{Lem:LB_Minimax}
Let $P_0, P_1\in \P$ satisfy $\ell(\theta(P_0), \theta(P_1))\geq 2\tau$. Then we have 
$$\mathcal R_n(\P, \ell, \eps, \delta)\geq \frac{\tau}{2}\left[1-\frac{1}{\sqrt{2}}\sqrt{ n \varphi(\eps, \delta)D_\kl(P_0\| P_1)}\right].$$
\end{lemma}

By comparing with the original non-private Le Cam's method \eqref{LeCam_KL}, we observe that the effect of $(\eps, \delta)$-LDP is to reduce the effective sample size from $n$ to $(1-e^{-\eps}(1-\delta))n$. Setting $\delta =0$, this result strengthens Duchi et al. \cite[Corollary 2]{Duchi_LDP_MinimaxRates}, where the effective sample size was shown to be $4\eps^2 n$ for sufficiently small $\eps$.

\begin{example}[(One-dimensional mean estimation)]
For some $k > 1$, we assume $\P$ is given by
$$\P = \P_k\coloneqq \{P\in \P(\X):~|\E_P[X]|\leq 1, \E_P[|X|^k]\leq 1\}.$$
The goal is to estimate $\theta(P) = \E_P[X]$ under $\ell = \ell_2^2$ the squared $\ell_2$ metric. This problem was first studied in  \cite[Propsition 1]{Duchi_LDP_MinimaxRates}  where it was shown $\mathcal R_n(\P_k, \ell_2^2, \eps, 0)\geq (n\eps^2)^{-(k-1)/k}$ \textit{only} for $\eps\leq 1$. 
Applying our framework to this example, we obtain a similar lower bound that holds for all $\eps\geq 0$ and $\delta\in [0,1]$. 
\begin{corollary}\label{corollary:LB_Pk}
For all $k> 1$, $\eps\geq 0$, and $\delta\in (0,1)$, we have
\begin{align}
\mathcal R_n(\P_k, \ell_2^2, \eps, \delta) &\gtrsim \min\Big\{1,\left[n\varphi^2(\eps, \delta)\right]^{-\frac{(k-1)}{k}}\Big\}.\label{LB_Example}
\end{align}
\end{corollary}
It is worth instantiating this corollary for some special values of $k$. Consider first the usual setting of finite variance setting, i.e., $k=2$. In the non-private case, it is known that the sample mean has mean-squared error that scales as $1/n$. According to Corollary~\ref{corollary:LB_Pk}, this rate worsens to $1/\varphi(\eps, \delta)\sqrt{n}$ in the presence of $(\eps, \delta)$-LDP requirement. As $k\to \infty$, the moment condition $\E_p[|X|^k]\leq 1$ implies the boundedness of $X$. In this case, Corollary~\ref{corollary:LB_Pk} implies the more standard lower bound $(\varphi^2(\eps, \delta)n)^{-1}$.      
\end{example}

\subsection{Locally Private Fano's Method}
Le Cam's method involves a pair of distributions $(P_0,P_1)$ in $\P$. However, it is possible to derive a stronger bound considering a larger subset of $\P$ by applying Fano's inequality  (see, e.g.,  \cite{Yu1997}).  We follow this path to obtain a better minimax lower bound for the non-interactive setting. 

Consider the index set $\V = \{1, \dots, |\V|\}$.  The non-private Fano's method relies on the Fano's inequality to write a lower bound for $\mathsf P_{\mathsf e}(V|X^n)$ in terms of mutual information as
\eqn{Eq:Fano}{\mathcal R_n(\P, \ell)\geq \tau \left[1-\frac{I(X^n; V) + \log 2}{\log |\V|}\right].} To incorporate privacy into this result, we need to derive an upper bound for $I(Z^n ; V)$ over all choices of mechanisms $\{\sK_i\}$. Focusing on the non-interactive mechanisms, the following lemma exploits Lemma~\ref{lemma:UB_Eta_Kn} for such an upper bound.

\begin{lemma}\label{Lemma_MI}
Given $X^n$ and $V$ as described above, let $Z^n$ be constructed by applying $\sK^{\otimes n}$ on $X^n$. If $\sK$ is $(\eps, \delta)$-LDP, then we have 
\begin{align*}
    I(Z^n; V) &\leq \varphi_n(\eps, \delta)I(X^n; V)\\
    &\leq \frac{n\varphi_n(\eps, \delta)}{|\V|^2}\sum_{v,v'\in \V}D_\kl(P_v\|P_{v'})
\end{align*}
\end{lemma}
This lemma can be compared with \cite[Corollary 1]{Duchi_LDP_MinimaxRates}, where it was  shown 
\begin{align}
    I(Z^n; V)
    &\leq 2(e^\eps-1)\frac{n}{|\V|^2}\sum_{v, v'\in \V}D_\kl(P_v\| P_{v'}).\label{Duchi_MI}
\end{align}
This is a looser bound than Lemma~\ref{Lemma_MI} for any $n\geq 1$ and $\eps\geq 0.4$ and only holds for $\delta = 0$.

\begin{example}[(High-dimensional mean estimation in an  $\ell_2$-ball)]For a parameter $r<\infty$, define 
\eqn{}{\P_{r}\coloneqq \{P\in \P(\mathsf{B}^d_2(r))\},}
where $\mathsf{B}^d_2(r)\coloneqq \{x\in \R^d:~\|x\|_2\leq r\}$ is the $\ell_2$-ball of radius $r$ in $\R^d$. The goal is to estimate the 
mean $\theta(P) = \E[X]$ given the private views $Z^n$.  This example was first studied in  \cite[Proposition 3]{Duchi_LDP_MinimaxRates} that states $\mathcal R_n(\P, \ell_2^2, \eps, 0)\gtrsim   r^2\min\left\{\frac{1}{\eps\sqrt{n}}, \frac{d}{n\eps^2}\right\}$ for $\eps\in (0,1)$. In the following, we use Lemma~\ref{Lemma_MI} to derive a similar lower bound for any $\eps\geq 0$ and $\delta\in (0,1)$, albeit slightly weaker than   \cite[Proposition 3]{Duchi_LDP_MinimaxRates}.

\begin{corollary}\label{corollary:Minimax_highD}
For the non-interactive setting, we have  
\eqn{}{\mathcal R_n(\P, \ell_2^2, \eps, \delta)\gtrsim   r^2\min\left\{\frac{1}{n\varphi_n(\eps, \delta)}, \frac{d}{n^2\varphi^2_n(\eps, \delta)}\right\}.} 
\end{corollary}
\end{example}

\section{Private Bayesian Risk}
In the minimax setting, the worst-case parameter is considered which usually leads to over-pessimistic bounds. In practice, the parameter that incurs a worst-case risk may appear with very small probability. To capture  this prior knowledge, it is reasonable to assume that the true parameter is sampled from an underlying prior distribution. In this case, we are interested in the \textit{Bayes risk} of the problem.

Let $\P=\{P_{X|\Theta}(\cdot|\theta):\theta\in \T\}$ be a collection of parametric probability distributions on $\X$ and the parameter space $\T$ is endowed with a prior $P_{\Theta}$, i.e., $\Theta\sim P_\Theta$. Given an i.i.d.\ sequence $X^n$ drawn from $P_{X|\Theta}$, the goal is to estimate $\Theta$ from a privatized sequence $Z^n$ via an estimator $\Psi:\Z^n\to \T$.  Here, we focus on the non-interactive setting.
Define the private Bayes risk as 
\eqn{}{R_n^{\mathsf{Bayes}}(P_\Theta, \ell, \eps, \delta) \coloneqq \inf_{\sK\in \Q_{\eps, \delta}}\inf_{\Psi}\E[\ell(\Theta, \Psi(Z^n))],} 
where the expectation is taken with respect to the randomness of both $\Theta$ and $Z^n$. 
It is evident that $R_n^{\mathsf{Bayes}}(P_\Theta, \ell, \eps, \delta)$ must depend on the prior $P_\Theta$. This dependence can be quantified by
\eqn{}{\L(\zeta)\coloneqq \sup_{t\in \T}\Pr(\ell(\Theta, t)\leq \zeta),}
for $\zeta<\sup_{\theta, \theta'\in \T} \ell(\theta, \theta')$. 
Xu and Raginsky \cite{Raginsky_ISIT_converses} showed that the non-private Bayes risk (i.e., $X^n = Z^n$), denoted by $R_n^{\mathsf{Bayes}}(P_\Theta, \ell)$, is lower bounded as 
\eqn{Xu_Maxim}{R_n^{\mathsf{Bayes}}(P_\Theta, \ell)\geq \sup_{\zeta>0} \zeta\left[1-\frac{I(\Theta; X^n) + \log 2}{\log(1/\L(\zeta))}\right].}
Replacing $I(\Theta; X^n)$ with $I(\Theta; Z^n)$ in this result and applying Lemma~\ref{lemma:UB_Eta_Kn} (similar to Lemma~\ref{Lemma_MI}), we 
can directly convert \eqref{Xu_Maxim} to a lower bound for $R_n^{\mathsf{Bayes}}(P_\Theta, \ell, \eps, \delta)$.
\begin{corollary}\label{Eq:Private_Raginsky}
In the non-interactive setting, we have 
\eq{R_n^{\mathsf{Bayes}}(P_\Theta, \ell, \eps, \delta)\geq \sup_{\zeta>0} \zeta\left[1-\frac{\varphi_n(\eps, \delta)I(\Theta; X^n) + \log 2}{\log(1/\L(\zeta))}\right].}
\end{corollary}
In the following theorem, we provide a lower bound for $R_n^{\mathsf{Bayes}}(P_\Theta, \ell, \eps, \delta)$ that directly involves $\sE_\gamma$-divergence, and thus leads to a tighter bounds than \eqref{Eq:Private_Raginsky}. For any pair of random variables $(A, B)\sim P_{AB}$ with marginals $P_A$ and $P_B$ and a constant $\gamma\geq 0$, we define their $E_\gamma$-information as 
$$I_\gamma(A; B)\coloneqq \sE_\gamma(P_{AB}\|P_AP_B).$$ 
\begin{theorem}\label{Thm:Bayesian_LB}
Let $\sK $ be an $(\eps, \delta)$-LDP mechanism. Then, for $n=1$ we have 
\eq{R_1^{\mathsf{Bayes}}(P_\Theta, \ell, \eps, \delta)\geq \sup_{\zeta>0}\zeta \left[1-\delta I_{e^\eps}(\Theta; X)-e^\eps\L(\zeta)\right],}
and for $n>1$ in non-interactive setting we have
\eq{R_n^{\mathsf{Bayes}}(P_\Theta, \ell, \eps, \delta)\geq \sup_{\zeta>0}\zeta \left[1- \varphi_n(\eps, \delta)I_{e^\eps}(\Theta; X^n)-e^\eps\L(\zeta)\right].}
\end{theorem}
We compare Theorem~\ref{Thm:Bayesian_LB} with  Corollary~\ref{Eq:Private_Raginsky} in the next example.
\begin{example}\label{Example_UniformTheta}
Suppose $\Theta$ is uniformly distributed on $[0,1]$,  $P_{X|\Theta=\theta}=\sBer(\theta)$, and $\ell(\theta, \theta')=|\theta-\theta'|$. As mentioned earlier, $\L(\zeta) \leq \min\{2\zeta, 1\}$. 
We can write for $\gamma = e^\eps$
\eqn{}{I_\gamma(\Theta; X^n)= \int_{0}^1 \sE_\gamma(P_{X^n|\theta}\|P_{X^n})\text{d}\theta.}
A straightforward calculation shows that $P_{X^n|\theta}(x^n) = \theta^{s(x^n)}(1-\theta)^{n-s(x^n)}$, for any $\theta\in [0,1]$, and $P_{X^n}(x^n) = \frac{s(x^n)!(n-s(x^n))!}{(n+1)!}$ where $s(x^n)$ is the number of 1's in $x^n$. 
Given these marginal and conditional distribution, one can obtain after algebraic manipulations
\eq{I_\gamma(\Theta; X^n)= \frac{1}{n+1}\sum_{s=0}^n\int_{0}^1 \left[\theta^{s}(1-\theta)^{n-s}\frac{(n+1)!}{s!(n-s)!} - \gamma \right]_+\text{d}\theta.}
Plugging this into Theorem~\ref{Thm:Bayesian_LB}, we arrive at a maximization problem that can be numerically solved.  Similarly, we compute $I(\Theta; X^n) = \int_{0}^1D_\kl(P_{X^n|\theta}\|P_{X^n})\text{d}\theta$ and plug it into Corollary~\ref{Eq:Private_Raginsky} and numerically solve the resulting optimization problem. In Fig.~\ref{fig:LDP_Uniform}, we compare these two lower bounds for $\delta = 10^{-4}$ and $n= 20$, indicating the advantage of Theorem~\ref{Thm:Bayesian_LB} for small $\eps$.
\begin{figure}
    \centering
    \includegraphics[scale=0.6]{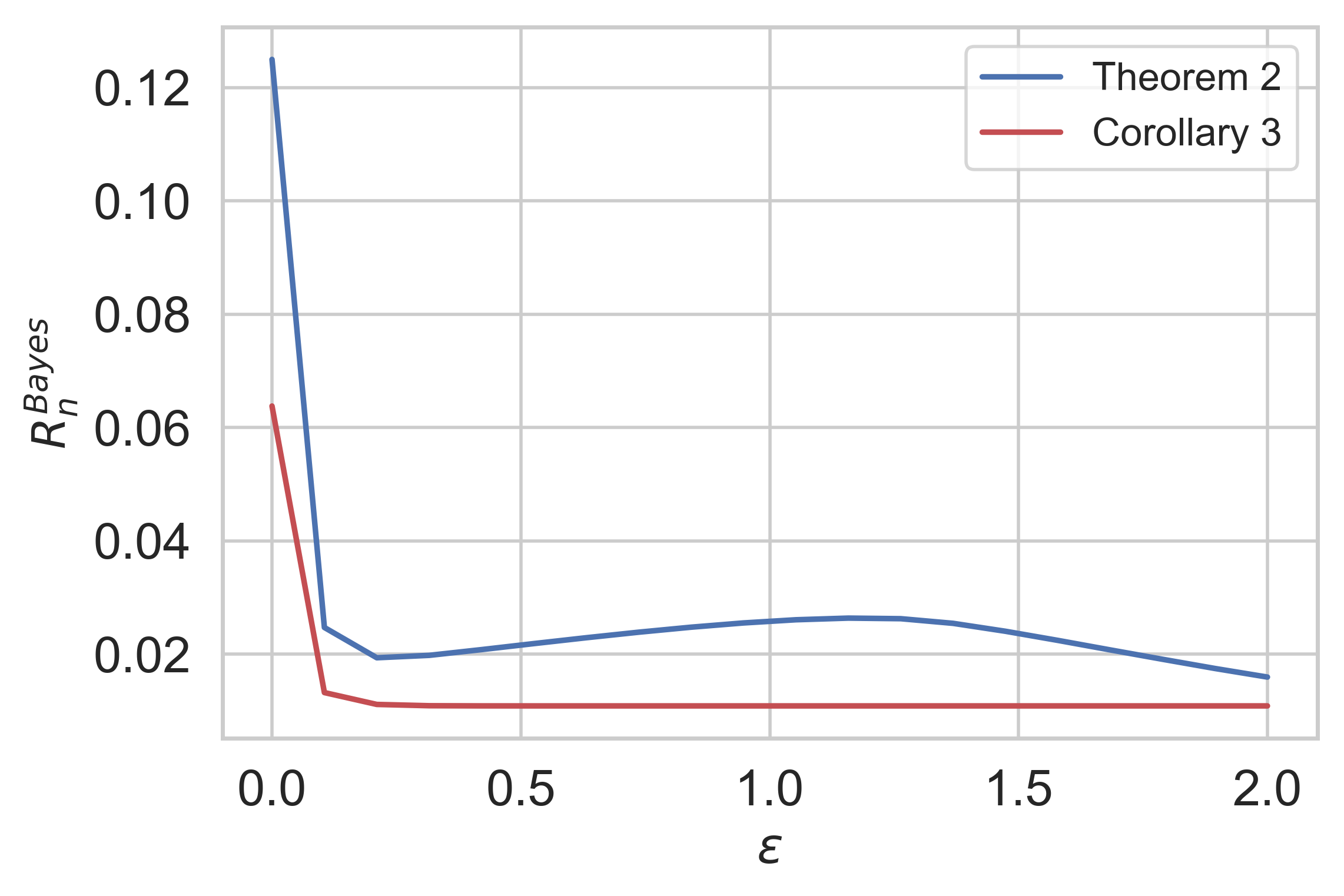}
    \caption{Comparison of the lower bounds obtained from Theorem~\ref{Thm:Bayesian_LB} and the private version of \cite[Theorem 1]{Raginsky_ISIT_converses} described in Corollary~\ref{Eq:Private_Raginsky} for Example~\ref{Example_UniformTheta} assuming $\delta = 10^{-4}$ and $n= 20$.}
    \label{fig:LDP_Uniform}
\end{figure}
\end{example}

\begin{remark}
The proof of Theorem~\ref{Thm:Bayesian_LB} leads to the following lower bound for the non-private Bayes risk
\eqn{Bayes_LB2}{R_n^{\mathsf{Bayes}}(P_\Theta, \ell)\geq \sup_{\substack{\zeta>0,\\ \gamma\geq0}}\zeta \left[1-I_{\gamma}(\Theta; X^n)-\gamma\L(\zeta)-(1-\gamma)_+\right].} For a comparison with \eqref{Xu_Maxim}, consider the following example.   
Suppose $\Theta$ is a uniform random variable on $[0,1]$ and $P_{X|\Theta = \theta}=\sBer(\theta)$. We are interested in the Bayes risk with respect to the $\ell_1$-loss function $\ell(\theta, \theta') =  |\theta-\theta'|$. 
It can be shown that $I(\Theta; X) = 0.19$ nats while 
\eqn{}{I_\gamma(\Theta; X) = \begin{cases} 0.25\gamma^2 & \text{if}~\gamma\in [0,1]\\
0.25(\gamma-2)^2 & \text{if}~\gamma\in [1,2]\\
0& \text{otherwise}.
\end{cases}
}
Moreover, $\L(\zeta) = \sup_{t\in [0,1]}\Pr(|\Theta-t|\leq \zeta)\leq \min\{2\zeta, 1\}$. It can be verified that \eqref{Xu_Maxim} gives $R_1^{\mathsf{Bayes}}(P_\Theta, \ell_1)\geq 0.03$, 
whereas our bound \eqref{Bayes_LB2} yields $R_1^{\mathsf{Bayes}}(P_\Theta, \ell_1)\geq 0.08$.
\end{remark}

\section{Private Hypothesis Testing}
We now turn our attention to the well-known problem of binary hypothesis testing under local differential privacy constraint. 
Suppose $n$ i.i.d. samples $X^n$ drawn from a distribution $Q\in \P(\X)$ are observed. Let now each $X_i$ be mapped to $Z_i$ via a mechanism $\sK_i\in \Q_{\eps, \delta}$ (i.e., sequential interaction is permitted). The goal is to distinguish between the null hypothesis $H_0: Q = P_0$ from the alternative $H_1: Q = P_1$ given $Z^n$.
Let $T$ be a binary statistic, generated from a randomized decision rule $P_{T|Z^n}:\Z^n \to \mathcal P(\{0,1\})$ where $1$ indicates that $H_0$ is rejected.
Type I and type II error probabilities corresponding to this statistic are given by $\Pr(T = 1|H_0)$ and $\Pr(T = 1|H_1)$, respectively. To capture the optimal trade-off between type I and type II error probabilities, it is customary to define 
$\beta^{\eps, \delta}_n(\alpha)\coloneqq \inf \Pr(T = 0|H_1)$ where the infimum is taken over all kernels $P_{T|Z^n}$ such that $\Pr(T = 1|H_0)\leq \alpha$ and non-interactive mechanisms $\sK^{\otimes n}$ with $\sK\in \Q_{\eps, \delta}$. In the following lemma, we apply Lemma~\ref{Lemma:UB_f_Div} to obtain an asymptotic lower bound for $\beta_n^{\eps, \delta}(\alpha)$.
  
\begin{corollary}\label{Cor:BHT}
We have for any $\eps\geq 0$ and $\delta\in [0,1]$
\eqn{}{
\lim_{n\to\infty}\frac{1}{n}\log\beta_n^{\eps, \delta}(\alpha)\geq -\varphi(\eps, \delta)D_\kl(P_0\|P_1).}
\end{corollary}
A similar result was proved by Kairouz et al. \cite[Sec.\ 3]{kairouz2014extremal_JMLR} that holds only for sufficiently ``small'' (albeit unspecified) $\eps$ and $\delta = 0$. When compared to Chernoff-Stein lemma \cite[Theorem 11.8.3]{cover2012elements}, establishing $D_\kl(P_0\|P_1)$ as the asymptotic exponential decay rate of $\beta_n^{\eps, \delta}(\alpha)$, the above corollary, once again, justifies the reduction of effective sample size from $n$ to $\varphi(\eps, \delta)n$ in the presence of $(\eps, \delta)$-LDP requirement.


\section{Mutual Information of LDP Mechanisms}
Viewing mutual information as a utility measure, we may consider maximizing mutual information under local differential privacy as yet another privacy-utility trade-off. To formalize this, let $X\sim P_X$. The goal is to characterize the supremum of $I(X; Z)$ over $\sK\in \Q_{\eps, \delta}$, i.e., the maximum information shared between $X$ and its $(\eps, \delta)$-LDP representation. Such mutual information bounds under local DP have appeared in the literature, e.g., McGregor et al. \cite{McGregor} provided a result that roughly states $I(X; Z)\leq 3\eps$ for $\sK\in \Q_{\eps}$ and Kairouz et al. \cite[Corollary 15]{kairouz2014extremal_JMLR} showed for sufficiently small $\eps$
\eqn{Kairouz_MI}{\sup_{\sK\in \Q_{\eps}}I(X; Z)\leq \frac{1}{2}P_X(A)(1-P_X(A))\eps^2,}
where $A\subset \X$ satisfies $A\in \argmin_{B\subset\X}|P_X(B)-\frac{1}{2}|$.
Next, we provide an upper bound for the mutual information under LDP that holds for \textit{all} $\eps\geq 0$ and $\delta\in [0,1]$.
\begin{corollary}\label{Cor:MI}
We have for any $\eps\geq 0$ and $\delta\in [0,1]$
\eqn{}{\sup_{\sK\in \Q_{\eps, \delta}} I(X; Z)\leq \varphi(\eps, \delta) H(X). }
\end{corollary}

\small{
\bibliographystyle{IEEEtran}
\bibliography{reference}
}
\normalsize
\appendix
We begin by some alternative expressions for $\sE_\gamma$-divergence that are useful for the subsequent proofs. 
It is straightforward to show that for any $\gamma\geq 0$, we have 
 \begin{align}
 	\sE_{\gamma}(P\|Q) &=\frac{1}{2}\int|\textnormal{d}P-\gamma \textnormal{d}Q| -\frac{1}{2}|1-\gamma| \label{eq:DefEgamma}\\
 	&= \sup_{A\subset \X} \left[P(A) - \gamma Q(A)\right]- (1-\gamma)_+\label{eq:DefEgamma2}\\
 	&=  P\big(\log\frac{\text{d}P}{\text{d}Q}>\log \gamma\big)-\gamma Q\big(\log\frac{\text{d}P}{\text{d}Q}>\log \gamma\big)\nonumber\\
 	&\quad - (1-\gamma)_+. \label{eq:EgammaPDif}
 \end{align}
 
The proof of Theorem~\ref{thm:LDP_Contraction} relies on the following theorem, recently proved by the authors in \cite[Theorem 3]{asoodeh2020online}. 

\begin{theorem}\label{Thm:Eta_Gama_formula}
For any $\gamma\geq 1$ and Markov kernel $\sK$ with input alphabet $\X$, we have
\eqn{pf_LDP1}{\eta_\gamma(\sK) = \sup_{x, x'\in \X} \sE_\gamma(\sK(\cdot|x)\|\sK(\cdot|x')).} 
\end{theorem}
Notice that for $\gamma = 1$, this theorem reduces to the well-known Dobrushin's result \cite{Dobrushin} that states 
\eqn{pf_LDP1_TV}{\eta_\tv(\sK) = \sup_{x, x'\in \X} \tv(\sK(\cdot|x),\sK(\cdot|x')).}

 \begin{proof}[Proof of Theorem~\ref{thm:LDP_Contraction}]
It follows from Theorem~\ref{Thm:Eta_Gama_formula} that 
\eqn{}{\eta_{e^\eps}(\sK)\leq \delta ~\Longleftrightarrow ~ \sup_{x, x'\in \X} \sE_{e^\eps}(\sK(\cdot|x)\|\sK(\cdot|x'))\leq \delta,}
which, according to \eqref{eq:DefEgamma2}, implies  
\eq{\eta_{e^\eps}(\sK)\leq \delta ~\Longleftrightarrow ~ \sup_{x, x'\in A}\sup_{A\subset\Z} \left[\sK(A|x)-e^\eps \sK(A|x')\right]\leq \delta.}
 Hence, in light of Definition~\ref{Def:LDP}, $\sK$ is $(\eps, \delta)$-LDP if and only if $\eta_{e^\eps}(\sK)\leq \delta$.  

 \end{proof}

 \begin{proof}[Proof of Lemma~\ref{Lemma:UB_f_Div}]
We first show the following upper and lower bounds for $\sE_\gamma$-divergence in terms of the total variation distance.

\noindent \textbf{Claim.} For any distributions $P$ and $Q$ on $\X$ and any $\gamma\geq 1$, we have 
\begin{equation}\label{Claim}
    1-\gamma(1-\tv(P\|Q))\leq \sE_\gamma(P\|Q)\leq \tv(P, Q).
\end{equation}
\begin{proof}[Proof of Claim] The upper bound is immediate from the definition of $\sE_\gamma$-divergence (and holds for any $\gamma\geq 0$). Note that 
\begin{align*}
 \gamma \tv(P\|Q) &  = \max_{A\subset \X}\left[\gamma P(A)-\gamma Q(A)\right]   \\
 &  =  \max_{A\subset \X}\left[P(A)-\gamma Q(A)+(\gamma-1) P(A)\right] \\
 &\leq \max_{A\subset \X}\left[P(A)-\gamma Q(A)\right] +(\gamma-1)\\
 & = \sE_\gamma(P\|Q) + \gamma - 1,
\end{align*}
where the last equality follows from \eqref{eq:DefEgamma2}. 
This immediately results in the lower bound in \eqref{Claim}.
\end{proof}

According to this claim, we can write for $\gamma\geq 1$
$$\tv(P,Q)\leq 1-\frac{1-\sE_\gamma(P\|Q)}{\gamma}.$$ 
Replacing $P$ and $Q$ with $\sK(\cdot|x)$ and $\sK(\cdot|x')$, respectively, for some $x$ and $x'$ in $\X$, we obtain 
$$\tv(\sK(\cdot|x),\sK(\cdot|x'))\leq 1-\frac{1-\sE_\gamma(\sK(\cdot|x)\|\sK(\cdot|x'))}{\gamma}.$$
Taking supremum over $x$ and $x'$ from both side and invoking Theorem~\ref{Thm:Eta_Gama_formula} and \eqref{pf_LDP1_TV}, we conclude that 
\eqn{Eta_TV_Eta_Gamma}{\eta_\tv(\sK)\leq 1-\frac{1-\eta_\gamma(\sK)}{\gamma}.}
It is known \cite{COHEN_Paper,Raginsky_SDPI} that for any Markov kernel $\sK$ and any convex functions $f$ we have 
\eqn{Eq:UBEtaTV}{\eta_f(\sK)\leq \eta_\mathsf{TV}(\sK),}
from which the desired result follows immediately.
\end{proof}

\begin{proof}[Proof of Lemma~\ref{lemma:UB_Eta_Kn}]
Given $n$ mechanisms $\sK_1, \sK_2, \dots, \sK_n$, we consider the non-interactive mechanism $P_{Z^n|X^n}$ given by 
$$P_{Z^n|X^n}(z^n|x^n) = \prod_{i=1}^n \sK_i(z_i|x_i).$$
If $\sK_i\in \Q_{\eps, \delta}$ for $i\in [n]$, then we have $\eta_{e^\eps}(\sK_i) \leq \delta$. According to \eqref{Eta_TV_Eta_Gamma}, it thus leads to 
$\eta_\tv(\sK_i)\leq \varphi(\eps, \delta)$. 
Invoking \cite[Corollary 9]{Polyankiy_SDPI_Networks} (see also \cite[Lemma 3]{Raginsky_ISIT_converses} and \cite[Eq. (62)]{Makur_SDPIjournal}), we obtain 
\begin{align*}
    \eta_\tv(P_{Z^n|X^n})&\leq \max_{i\in [n]}\left[1-(1-\eta_{\tv}(\sK_i))^n\right]\\
    &\leq \left[1-(1-\varphi(\eps, \delta))^n\right]\\
    & = \varphi_n(\eps, \delta).
\end{align*}

\end{proof}

\begin{proof}[Proof of Lemma~\ref{Lem:LB_Minimax}]
Recall that $X^n$ is an i.i.d.\ sample of distribution $P$ and each $Z_i$, $i\in [n]$ is obtained by applying $\sK_i$ to $X_i$. Note that by assumption $\sK_i$ specifies the conditional distribution  $P_{Z_i|X_i, Z^{i-1}}$. Let $M^n_0$ and $M^n_1$ denote the distribution of $Z^n$ when $P = P_0$ and $P = P_1$, respectively.  Thus, we have for $P= P_0$ and any $z^n\in \Z^n$
\aln{M^n_0(z^n) & = \prod_{i=1}^n P_{Z_i|Z^{i-1}}(z_i|z^{i-1}) \\
&= \prod_{i=1}^n \left(P_{X_i|Z^{i-1}=z^{i-1}}\sK_i\right)(z_i) \\
&= \prod_{i=1}^n \left(P_{0}\sK_i\right)(z_i).}
Having this in mind, we can write  
\aln{
\tv^2(M^n_0, M^n_1) &\stackrel{(a)}{\leq} \frac{1}{2}D_\kl(M^n_0\| M^n_1)\\
& \stackrel{(b)}{=}  \frac{1}{2}\sum_{i=1}^n D_\kl(P_0\sK_i\| P_1\sK_i)\\
& \stackrel{(c)}{\leq} \frac{1}{2}\sum_{i=1}^n \varphi(\eps, \delta)D_\kl(P_0\| P_1)\label{KL_Proof1}
}
where $(a)$ follows from Pinsker's inequality, $(b)$ is due to the chain rule of KL divergence, $(c)$ is an application of Lemma~\ref{Lemma:UB_f_Div}. Plugging \eqref{KL_Proof1} in \eqref{LeCam_TV}, we obtain the desired result. 
\end{proof}
 
 \begin{proof}[Proof of Corollary~\ref{corollary:LB_Pk}]
 
 Fix $\omega\in (0,1]$ and consider two distributions $P_0$ and $P_1$ on $\{-\omega^{-\frac{1}{k}}, 0, \omega^{-\frac{1}{k}}\}$ defined as
$$P_0(-\omega^{-\frac{1}{k}}) = \omega, \qquad P_0(0) = 1-\omega,$$
and 
$$P_1(\omega^{-\frac{1}{k}}) = \omega, \qquad P_1(0) = 1-\omega.$$
It can be verified that both $P_0$ and $P_1$ belong to $\P_k$.  Note that $\ell_2^2(\theta(P_0), \theta(P_1)) = 2\omega^{\frac{2(k-1)}{k}}$.
Let $M_0^n = P^{\otimes n}_0\sK^n$ and $M_1^n=P^{\otimes n}_1\sK^n$ be the corresponding output distributions of the mechanism $\sK^n=\sK_1\dots\sK_n$, the composition of mechanisms $\sK_i$. Le Cam's bound for $\ell_2^2$-metric yields   
\begin{align}
    \mathcal R_n(\P_k, \ell_2^2, \eps, \delta)&\geq \omega^{\frac{2(k-1)}{k}}(1-\tv(M^n_0, M^n_1)\nonumber\\
    &\geq \omega^{\frac{2(k-1)}{k}}\left(1-H(M^n_0, M^n_1)\right),\label{Hell1}
\end{align}
where the last inequality follows from the fact $\tv(P,Q)\leq H(P,Q)$ for $H(P, Q)$ being the Hellinger distance. Notice that $M_0^n = \prod_{i=1}^n(P_0\sK_i)$ and $M_1^n = \prod_{i=1}^n(P_1\sK_i)$ where each $\sK_i$ for $i\in [n]$ is $(\eps, \delta)$-LDP. 
It is well known that 
$$H^2\left(\prod_{i=1}^nP_i, \prod_{i=1}^nQ_i\right) = 2-2\prod_{i=1}^n\left(1-\frac{1}{2}H^2(P_i, Q_i)\right).$$
Thus, 
\begin{align}
    H^2(M_0^n, M_1^n)&= 2-2\prod_{i=1}^n\left(1-\frac{1}{2}H^2(P_0\sK_i, P_1\sK_i)\right)\nonumber\\
&\leq 2-2\prod_{i=1}^n\left(1-\frac{\varphi(\eps, \delta)}{2}H^2(P_0, P_1)\right)\nonumber\\
&= 2-2\left(1-\frac{\varphi(\eps, \delta)}{2}H^2(P_0, P_1)\right)^n\nonumber\\
&= 2-2\left(1-\omega \varphi(\eps, \delta)\right)^n\label{Hell2}
\end{align}
Hence, we obtain 
\eqn{TV_UB_EXample}{\tv(M_0^n, M_1^n)\leq \sqrt{2-2\left(1-\omega \varphi(\eps, \delta)\right)^n}.}
Plugging \eqref{Hell2} into \eqref{Hell1}, we obtain 
\eqn{}{\mathcal R_n(\P_k, \ell_2^2, \eps, \delta)\geq \omega^{\frac{2(k-1)}{k}}\left[1-\sqrt{2}\sqrt{1-(1-\omega\varphi(\eps,\delta))^n}\right].}
Now, choose
$\omega = \min\left\{1, \frac{1}{\varphi(\eps, \delta)}\left[1-\left(\frac{7}{8}\right)^{\frac{1}{\sqrt{n}}}\right]\right\}$. Notice that we assume $\delta>0$ and hence $\varphi(\eps, \delta)>0$ regardless of $\eps$. Plugging this choice of $\omega$ into the above bound, we obtain 
\begin{align}
\mathcal R_n(\P_k, \ell_2^2, \eps, \delta) &\gtrsim (\varphi(\eps, \delta))^{-\frac{2(k-1)}{k}}\left[1-\left(\frac{7}{8}\right)^{\frac{1}{\sqrt{n}}}\right]^{\frac{2(k-1)}{k}}\nonumber\\
&\gtrsim (\varphi(\eps, \delta))^{-\frac{2(k-1)}{k}} n^{-\frac{k-1}{k}}.
\end{align}
  \end{proof}
 
\begin{proof}[Proof of Lemma~\ref{Lemma_MI}]
Note that we have the Markov chain $V\markov X^n\markov Z^n$. It has been shown in \cite[Problem 15.12]{csiszarbook} (see also \cite{Anantharam_SDPI,csiszarbook}) that for any channel $P_{B|A}$ connecting random variable $A$ to $B$, we have
\eqn{Eta_KL_MI}{\eta_\kl(P_{B|A}) = \sup_{\substack{P_{AU}:\\U\markov A\markov B}}\frac{I(U; B)}{I(U; A)}.}
Replacing $A$ and $B$ with $X^n$ and $Z^n$, respectively, in the above equation, we obtain
\begin{align*}
I(Z^n; V)&\leq \eta_\kl(\sK^{\otimes n})I(X^n; V)\\
&=  \eta_\kl(\sK^{\otimes n}) \frac{n}{|\V|}\sum_{V=v}D_\kl(P_v\|\bar P),
\end{align*}
where $\sK^{\otimes n} = P_{Z^n|X^n}$ and $\bar P = \frac{1}{|\V|}\sum_{v\in \V}P_v$. The desired result then follows from Lemma~\ref{lemma:UB_Eta_Kn} and the convexity of KL-divergence. 
\end{proof}

\begin{proof}[Proof of Corollary~\ref{corollary:Minimax_highD}]
The proof strategy is as follows: we first construct a set of probability distribution $\{P_v\}$ for $v$ taking values in a finite set $\V$ and then apply Fano's inequality \eqref{Eq:Fano} where $V$ is a uniform random variable on $\V$.  Duchi el el. \cite[Lemma 6]{Duchi_LDP_MinimaxRates} showed that there exists a set $\V_k$ of the $k$-dimensional hypercube $\{-1, +1\}^k$ satisfying $\|v-v'\|_1\geq \frac{k}{2}$ for each $v, v'\in \V_k$ with $v\neq v'$ and some integer $k\in [d]$, while $|\V_k|$ being at least $\lceil e^{\frac{k}{16}}\rceil$. 
If $k<d$, one can extend $\V_k\subset \R^k$ to a subset of $\R^d$ by considering $\V = \V_k\times \{0\}^{d-k}$. 
Fix $\omega \in (0,1)$ and define a distribution $P_v\in \P(\mathsf{B}^d_2(r))$ for $v\in \V$ as follows: Choose an index $j\in [k]$ uniformly and set $P_{v}(r \mathsf{e}_j) = \frac{1+\omega v_j}{2}$ and $P_{v}(-r \mathsf{e}_j) = \frac{1-\omega v_j}{2}$ where $\mathsf{e}_j$ is the standard basis vector in $\R^d$. Given $v\in \V_k$, let $X\sim P_v$ be a random variable taking values in $\{\pm r \mathsf{e}_j\}_{j=1}^k$ and $X^n$ be an i.i.d.\ sample of $X$. Furthermore, as before, let $Z^n$ be a privatized sample of $X^n$ obtained by $\sK^{\otimes n}$ with $\sK$ being an $(\eps, \delta)$-LDP mechanism. To apply Fano's inequality, we first need to bound $I(Z^n; V)$. According to Lemma~\ref{Lemma_MI}, we have 
\begin{align}
  I(Z^n; V)&\leq \varphi_n(\eps, \delta)\frac{n}{|\V_k|^2}\sum_{v, v'}\kl(P_v\|P_{v'})\nonumber\\
  &\leq \varphi_n(\eps, \delta)I(X^n; V).\label{Example2_UB_MI}
  \end{align}
Hence, bounding $I(Z^n; V)$ reduces to bounding $I(X^n; V)$. To this end, first notice that $I(X^n; V)\leq n I(X; V)$. Let $K$ be a uniform random variable on $[k]$ independent of $V$ that chooses the coordinate of $V$. 
Note that $K$ can be determined by $X$ and hence 
\begin{align*}
    I(X; V) &= I(X, K; V) \\
    &= I(X; V|K) \\
    & \leq \log2-h_{\mathsf{b}}(\frac{1-\omega}{2})\\
    &\leq \omega\log 2,
\end{align*}
where the last inequality follows from the fact that $h_{\mathsf{b}}(a)\geq 2a\log2$ for $a\in [0, \frac{1}{2}]$ due to the concavity of entropy. 
Consequently, we can write 
\eqn{Example2_UB_MI2}{I(Z^n; V)\leq n\omega\varphi_n(\eps, \delta)\log 2 .}
Applying Fano's inequality, we obtain \aln{\mathcal R_n(\P, \ell, \eps, \delta)&\geq \frac{r^2\omega^2}{k} \left[1-\frac{(1+n\omega\varphi_n(\eps, \delta))\log 2}{\log |\V|}\right]\\
& \geq   \frac{r^2\omega^2}{k} \left[1-\frac{16(1+n\omega\varphi_n(\eps, \delta))\log 2}{k}\right].}
Setting $\omega = \min\{1, \frac{k}{50n\varphi_n(\eps, \delta)}\}$ and assuming $k\geq 16$, we can write
\eqn{}{\mathcal R_n(\P, \ell, \eps, \delta)\gtrsim   r^2\max_{k\in [d]}\min\left\{\frac{1}{k}, \frac{k}{n^2\varphi^2_n(\eps, \delta)}\right\} .}
By choosing $k = \min\{n\varphi_n(\eps, \delta), d\}$, we obtain 
\eqn{}{\mathcal R_n(\P, \ell, \eps, \delta)\gtrsim   r^2\min\left\{\frac{1}{n\varphi_n(\eps, \delta)}, \frac{d}{n^2\varphi^2_n(\eps, \delta)}\right\}.} 
\end{proof}

\begin{proof}[Proof of Theorem~\ref{Thm:Bayesian_LB}]
Let $\hat \Theta = \Psi(Z^n)$ be an estimate of $\Theta$ for some $\Psi$ and $p_\zeta\coloneqq P_{\Theta\hat \Theta}(\ell(\Theta, \hat \Theta)\leq \zeta)$ and $q_\zeta\coloneqq (P_{\Theta}P_{\hat \Theta})(\ell(\Theta, \hat \Theta)\leq \zeta)$, i.e., $p_\zeta$ and $q_\zeta$ correspond to the probability of the event $\{\ell(\Theta, \hat\Theta)\leq \zeta\}$ under the joint and product distributions, respectively. By definition, we have for any $\gamma\geq 1$
\al{I_\gamma(\Theta; \hat \Theta) &= 
\sE_\gamma(P_{\Theta\hat\Theta}\| P_\Theta P_{\hat\Theta}) \\ 
&=\sup_{A\subset \T\times \T} \left[P_{\Theta\hat \Theta}(A) - \gamma (P_\Theta P_{\hat \Theta})(A)\right]\\
&\geq p_\zeta - \gamma q_\zeta\\
&\geq p_\zeta-\gamma\L(\zeta),}
where the last inequality follows from the fact that $q_\zeta\leq \L(\zeta)$, that can be shown as follows
\al{q_\zeta &= \int_{\T}\int_{\T}1_{\{\ell(\theta, \hat\theta)\leq\zeta\}}P_{\Theta}(\text{d}\theta) P_{\hat\Theta}(\text{d}\hat\theta)\\
&\leq \sup_{t\in \T}\int_{\T}1_{\{\ell(\theta, t)\leq\zeta\}}P_{\Theta}(\text{d}\theta)\\
& = \L(\zeta).}
Recalling that $\Pr(\ell(\Theta, \hat \Theta)>\zeta) = 1-p_\zeta$, the above thus implies 
\eqn{}{\Pr(\ell(\Theta, \hat \Theta)>\zeta)\geq 1-I_\gamma(\Theta; \hat \Theta)-\gamma\L(\zeta).}
Since, by Markov's inequality, $\E[\ell(\Theta, \hat \Theta)]\geq \zeta\Pr(\ell(\Theta, \hat \Theta)\geq \zeta)$, we can write by setting $\gamma = e^\eps$
\al{R_n^{\mathsf{Bayes}}(P_\Theta, \ell, \eps, \delta)&\geq \zeta \left[1-I_{e^\eps}(\Theta; \hat \Theta)-e^\eps\L(\zeta)\right]\\
&\geq \zeta \left[1-I_{e^\eps}(\Theta; Z^n)-e^\eps\L(\zeta)\right],}
where the second inequality comes from the data processing inequality for $I_\gamma$. To further lower bound the right-hand side, we write
\al{I_{e^\eps}(\Theta; Z^n) &= \int_{\T} \sE_{e^\eps}(P_{Z^n|\Theta=\theta}\|P_{Z^n})P_{\Theta}(\text{d}\theta)\\
& \leq \eta_{e^\eps}(P_{Z^n|X^n})\int_{\T} \sE_{e^\eps}(P_{X^n|\Theta=\theta}\|P_{X^n})P_{\Theta}(\text{d}\theta)\\
& = \eta_{e^\eps}(P_{Z^n|X^n}) I_{e^\eps}(\Theta; X^n),}
where the inequality follows from the definition of contraction coefficient.
When $n=1$, we have $\eta_{e^\eps}(\sK)\leq \delta$ as $\sK=P_{Z|X}$ is assumed to be $(\eps, \delta)$-DP.  
For $n>1$, we invoke Lemma~\ref{lemma:UB_Eta_Kn} to obtain $\eta_{e^\eps}(P_{Z^n|X^n}) \leq \varphi_n(\eps, \delta)$.
\end{proof}

\begin{proof}[Proof of Corollary~\ref{Cor:BHT} ]
Let $\beta_n(\alpha)\coloneqq \beta^{\infty, 1}_n(\alpha)$ be the non-private trade-off between type I and type II error probabilities (i.e., $Z^n = X^n$). According to Chernoff-Stein lemma (see, e.g., \cite[Theorem 11.8.3]{cover2012elements}), we have
\eqn{Chernoff}{\lim_{n\to \infty}\frac{1}{n}\log\beta_n(\alpha) = -D_\kl(P_0\|P_1).}
Assume now that, $Z^n$ is the output of  $\sK^{\otimes n}$ for an $(\eps, \delta)$-LDP mechanism $\sK$.  
According to \eqref{Chernoff}, we obtain that
\eqn{Chernoff2}{\lim_{n\to \infty}\frac{1}{n}\log\beta^{\eps, \delta}_n(\alpha) = -\sup_{\sK\in \Q_{\eps, \delta}}D_\kl(P_0\sK\|P_1\sK).}
Applying Lemma~\ref{Lemma:UB_f_Div}, we obtain the desired result. 
\end{proof}

\begin{proof}[Proof of Corollary~\ref{Cor:MI}]
Consider the (trivial) Markov chain $X\markov X\markov Z$. According to \eqref{Eta_KL_MI}, we can write $I(X; Z)\leq \eta_\kl(\sK)H(X)$. The desired result then immediately follows from Lemma~\ref{Lemma:UB_f_Div}.
\end{proof}

\end{document}